\documentclass[a4paper,UKenglish]{lipics-v2018}
\usepackage{microtype}
\usepackage[T1]{fontenc}
\usepackage[utf8]{inputenc}
\usepackage{graphicx}
\usepackage{amsmath}
\usepackage{amssymb}
\usepackage{xspace}
\usepackage{color}
\usepackage[linesnumbered,ruled,vlined]{algorithm2e}

\nolinenumbers
\usepackage{tikz,fp,ifthen}

\newcommand{\probl}{\text{GBSM}\xspace}
\newcommand{\greedymaxcover}{\ensuremath{\mathtt{GreedyMaxCover}}\xspace}

\title{
Generalized budgeted submodular set function maximization
}

\titlerunning{Generalized budgeted submodular set function maximization}

\author{Francesco Cellinese}{Gran Sasso Science Institute, L'Aquila,  Italy.}{francesco.cellinese@gssi.it}{}{}

\author{Gianlorenzo D'Angelo}{Gran Sasso Science Institute, L'Aquila,  Italy.}{gianlorenzo.dangelo@gssi.it}{}{}

\author{Gianpiero Monaco}{University of L’Aquila, L'Aquila, Italy.}{gianpiero.monaco@univaq.it}{}{}

\author{Yllka Velaj}{University of Chieti-Pescara, Pescara, Italy.}{yllka.velaj@unich.it}{}{}

\authorrunning{F. Cellinese, G. D'Angelo, G. Monaco, and Y. Velaj} 

\Copyright{Francesco Cellinese, Gianlorenzo D'Angelo, Gianpiero Monaco, and Yllka Velaj}

\subjclass{
, 
}


\EventEditors{Igor Potapov, Paul Spirakis, and James Worrell}
\EventNoEds{3}
\EventLongTitle{43rd International Symposium on Mathematical Foundations of Computer Science (MFCS 2018)}
\EventShortTitle{MFCS 2018}
\EventAcronym{MFCS}
\EventYear{2018}
\EventDate{August 27--31, 2018}
\EventLocation{Liverpool, GB}
\EventLogo{}
\SeriesVolume{117}
\ArticleNo{31} 

\begin{document}
\maketitle

\begin{abstract}
In this paper we consider a generalization of the well-known budgeted maximum coverage problem. 
We are given a ground set of elements and a set of bins. The goal is to find a subset of elements along with an associated set of bins, such that the overall cost is at most a given budget, and the profit is maximized. 
Each bin has its own cost and the cost of each element depends on its associated bin. The profit is measured by a monotone submodular function over the elements.

We first present an algorithm that guarantees an approximation factor of $\frac{1}{2}\left(1-\frac{1}{e^\alpha}\right)$, where $\alpha \leq 1$ is the approximation factor of an algorithm for a sub-problem. We give two polynomial-time algorithms to solve this sub-problem. The first one gives us $\alpha=1- \epsilon$ if the costs satisfies a specific condition, which is fulfilled in several relevant cases, including the unitary costs case and the problem of maximizing a monotone submodular function under a knapsack constraint. The second one guarantees $\alpha=1-\frac{1}{e}-\epsilon$ for the general case. 
The gap between our approximation guarantees and the known inapproximability bounds is $\frac{1}{2}$.

We extend our algorithm to a bi-criterion approximation algorithm in which we are allowed to spend an extra budget up to a factor $\beta\geq 1$ to guarantee a $\frac{1}{2}\left(1-\frac{1}{e^{\alpha\beta}}\right)$-approximation. 
If we set $\beta=\frac{1}{\alpha}\ln \left(\frac{1}{2\epsilon}\right)$, the algorithm achieves an approximation factor of $\frac{1}{2}-\epsilon$, for any arbitrarily small $\epsilon>0$.
\end{abstract}

\section{Introduction}\label{sec:introduction}

The Maximum Coverage (MC) is a fundamental combinatorial optimization problem which has several applications in job scheduling,
facility locations and resource allocations \cite[Ch. 3]{H97}, as well as in influence maximization \cite{KKT15}. In the classical definition we are given a ground set $X$, a collection $S$ of subsets of $X$ with unit cost, and a budget $k$. The goal is selecting a subset $S' \subseteq S$, such that $|S'| \leq k$, and the number of elements of $X$ covered by $S'$ is maximized. 
A natural greedy algorithm starts with an empty solution and iteratively adds a set with maximum number of uncovered elements until $k$ sets are selected.
This algorithm has an approximation of $1-\frac{1}{e}$~\cite{NWF78} and such result is tight given the inapproximability result due to Feige \cite{F98}. An interesting special case of the problem where this inapproximability result does not hold is when the size of the sets in $S$ is small. In the maximum $h$-coverage, $h$ denotes the maximum size of each set in $S$. This problem is APX-hard for any $h \geq 3$ \cite{K91} (notice that when $h=2$ it is the maximum matching problem), while a simple polynomial local search heuristic has an approximation ratio very close to $\frac{2}{h}$ \cite{C09}. A polynomial time algorithm with approximation factor of $\frac{5}{6}$ is possible for the case when $h=3$ \cite{CM13}. In the Budgeted Maximum Coverage (BMC) problem, which is an extension of the maximum coverage, the cost of the sets in $S$ are arbitrary, and thus a solution is feasible if the overall cost of the selected subset $S' \subseteq S$ is at most $k$. In \cite{khuller1999budgeted}, the authors present a polynomial time (greedy) algorithm with approximation factor of $1-\frac{1}{e}$. In the Generalized Maximum Coverage (GMC) problem every set $s \in S$ has a cost $c(s)$, and every element $x \in X$ has a different weight and cost that depend on which set covers it. In \cite{CK08}, a polynomial time (greedy) algorithm with approximation factor of $1-\frac{1}{e} - \epsilon$, for any $\epsilon >0$, has been shown.

In all the above problems the profit of a solution is given by the sum of the weights of the covered elements. An important and studied extension is adopting a nonnegative, nondecreasing, submodular function $f$, which assigns a profit to each subset of elements. In the Submodular set Function subject to a Knapsack Constraint maximization (SFKC) problem we have a cost $c(x)$ for any element $x \in X$, and the goal is selecting a set $X' \subseteq X$ of elements that maximizes $f(X')$, where $f$ is a monotone submodular function subject to the constraint that the sum of the costs of the selected elements is at most $k$. This problem admits a polynomial time algorithm that is $\left(1-\frac{1}{e}\right)$-approximation \cite{S04}. Since the MC problem is a special case of SFKC problem, such result is tight.
A more general setting was considered in \cite{IB13}, where the authors consider the following problem called Submodular Cost Submodular Knapsack (SCSK): given a set of elements $V= \{1,2,\ldots,n\}$, two monotone non-decreasing submodular functions $g$ and $f$ ($f,g : 2^{V} \rightarrow \mathbb{R}$), and a budget $b$, the goal is finding a set of elements $X \subseteq V$ that maximizes the value $g(X)$ under the constraint that $f(X) \leq b$. They show that the problem cannot be approximated within any constant bound. Moreover, they give a $1/n$ approximation algorithm and mainly focus on bi-criterion approximation. 

In this paper we consider the Generalized Budgeted submodular set function Maximization problem (\probl) that is not captured by any of the above settings. We are given a ground set of elements $X$, a set of bins $S$, and a budget $k$. The goal is to find a subset of elements along with an associated set of bins such that the overall costs of both is at most a given budget and the profit is maximized. Each bin has its own cost, while the cost of each element depends on its associated bin. Finally, the profit is measured by a monotone submodular function over the elements. 

We emphasize that the problem considered here is not a special case of the GMC problem, since we consider any monotone submodular functions for the profits. Moreover, we now show that our cost function is not submodular and thus that the setting (SCSK) considered in \cite{IB13} does not generalize our model. Given a set of elements $X$, we denote with $c(X)$ the minimum cost of covering them, that is the best choice of the bins able to cover the elements of $X$ with the minimum cost. Consider the instance $(S,X) = (\{s_1,s_2\},\{x_1,x_2,x_3\} )$ with $c(s_1)=c(s_2)=1$ and the costs of associating elements to bins depicted in Table~\ref{fig:apx:cost}, where $M$ is a large positive value.  Let $T=\{x_1,x_2\}$ and $T \supseteq S=\{x_1\}$. We now show that $c(S \cup \{x_3\}) - c(S) < c(T \cup \{x_3\}) - c(T)$. In fact, notice that $c(S)= 2-\epsilon$, that is covering the element $x_1$ with bin $s_2$. Moreover, $c(T)= 3$, that is covering the elements $x_1$ and $x_2$ with bin $s_1$. Finally, $c(S \cup \{x_3\})=1+1-\epsilon + \epsilon = 2$ (i.e. $x_1$ and $x_3$ are associated to $s_2$) and $c(T \cup \{x_3\})=1+1+1+1-\epsilon + \epsilon$ (i.e. $x_1$ and $x_3$ are associated to $s_2$, while $x_2$ is associated to $s_1$). We conclude that $c(S \cup \{x_3\}) - c(S)= \epsilon < c(T \cup \{x_3\}) - c(T) = 1$.
 
\begin{table}
\centering
\caption{}\label{fig:apx:cost}
\begin{tabular}{c c}
\hline
 $S, X$ &   $c(S, X)$\\
\hline
$(s_1, x_1)$ & $1$\\
$(s_1, x_2)$ & $1$\\
$(s_1, x_3)$ & $M$\\
$(s_2, x_1)$ & $1-\epsilon$\\
$(s_2, x_2)$ & $M$\\
$(s_2, x_3)$ & $\epsilon$\\
\hline
\end{tabular}
\end{table}

Finally, we notice that our setting extends the SFKC problem, given that, the cost of an element is not fixed like in SFKC, but instead depends on the bin used for covering it. 

In addition to its theoretical appeal, our setting is motivated by the adaptive seeding problem, which is an algorithmic challenge motivated by influence maximization in social networks~\cite{BPRSS16,SS13}. In its non-stochastic version, the problem is to select amongst certain accessible nodes in a network, and then select amongst neighbors of those nodes, in order to maximize a global objective function. In particular, given a set $X$ and its neighbors $N(X)$ there is a monotone submodular function defined on $N(X)$, and the goal is to select $t \leq k$ elements in $X$ connected to a set of size at most $k-t$ for which the submodular function has the largest value. Our setting is an extension of it since we consider more general costs.

\subsection{Our results}

In Section~\ref{sec:general} we present an algorithm that guarantees an approximation factor of $\frac{1}{2}\left(1-\frac{1}{e^\alpha}\right)$ for \probl. Here, $\alpha$ is the approximation factor of an algorithm used to select a subset of elements whose ratio between marginal increment in the objective function and marginal cost is maximum. 
We give two polynomial-time algorithms to solve this sub-problem. 
In particular, in Section~\ref{sec:epsilonblock} we propose an algorithm that gives us $\alpha=1- \epsilon$ if the costs satisfy a specific condition.  This latter is fulfilled in several relevant cases including the unitary costs case and the problem of maximizing a monotone submodular function under a knapsack constraint. In Section~\ref{sec:slot} we propose an algorithm that guarantees $\alpha=1-\frac{1}{e}-\epsilon$ for the general case. 

The gap between our approximation guarantees and the known inapproximability bounds, i.e. the $1-\frac{1}{e}$ hardness for the MC problem~\cite{F98}  and the $1-\frac{1}{e^{1-\frac{1}{e}}}$ hardness for the non-stochastic adaptive seeding problem with knapsack constraint~\cite{RSS15}, is $\frac{1}{2}$, unless $P=NP$.
%
%
%

In Section~\ref{sec:bicriteria}, we extend our algorithm to a bi-criterion approximation algorithm in which we are allowed to spend an extra budget up to a factor $\beta$.
An algorithm gives a $[\rho, \beta]$ bi-criterion approximation for \probl if it is guaranteed to obtain a solution $(S',X')$ such that $f(X')\geq \rho f(X^*)$ and $c(S', X')\leq\beta k$, where $X^*$ is the optimal solution. 
We denote by $\beta$ the extra-budget we are allowed to use in order to obtain a better approximation factor.
Our algorithm guarantees a $\left[\frac{1}{2}\left(1-\frac{1}{e^{\alpha\beta}}\right),\beta\right]$-approximation. If we set $\beta=\frac{1}{\alpha}\ln \left(\frac{1}{2\epsilon}\right)$, the algorithm achieves an approximation factor of $\frac{1}{2}-\epsilon$, for any arbitrarily small $\epsilon>0$.

\subsection{Related work}
Maximum coverage and submodular set function maximization are important problems. In the literature, besides the above mentioned ones, there are many other papers dealing with related issues. For instance, in the maximum coverage with group budgeted constraints, the set $S$ is partitioned into groups, and the goal is to pick $k$ sets from $S$ to maximize the cardinality of their union with the restriction that at most one set can be picked from each group. In \cite{CK04}, the authors propose a $\frac{1}{2}$-approximation algorithms for this problem, and smaller constant approximation algorithm for the cost version. In the ground-set-cost budgeted maximum coverage problem, given a budget and a hypergraph, where each vertex has a non-negative cost and a non-negative profit, we want to select a set of hyperedges such that the total cost of the covered vertices is at most the budget and the total profit of all covered vertices is maximized. This problem is strictly harder than budgeted max coverage. The difference of our problem to the budgeted maximum coverage problem is that the costs are associated with the covered vertices instead of the selected hyperedges. In \cite{VKS16}, the authors obtain a $\frac{1}{2}\left(1 - \frac{1}{\sqrt e}\right)$-approximation algorithm for graphs (which means having sets of size 2) and an FPTAS if the incidence graph of the hypergraph is a forest (i.e. the hypergraph is Berge-acyclic). 

Maximizing submodular set function is another important research topic. The general version of the problem is: given a set of elements and a monotone submodular function, the goal is to find the subset of elements that gives the maximum value, subjected to some constraints. The case when the subset of elements must be an independent set of the matroid over the set of elements has been considered in \cite{CCPV11}, where the authors show an optimal randomized $\left(1- \frac{1}{e}\right)$-approximation algorithm. A simpler algorithm has been proposed in \cite{FW14}. The case of multiple $k$ matroid constraints has been considered in \cite{LSV10}, where the authors give a $\frac{1}{k+\epsilon}$-approximation. An improved result appeared in \cite{W12}. Finally, unconstrained (resp. constrained) general non-monotone submodular maximization, have been considered in \cite{BFNS12,FMV07} (resp. \cite{VCZ11}).          

Another related topic is the adaptive seeding problem in which the aim is to select amongst a set $X$ of nodes of a network, called the \emph{core}, and then adaptively selecting amongst the neighbors $N(X)$ of those nodes as they become accessible in order to maximize a submodular function of the selected nodes in $N(X)$~\cite{BPRSS16,SS13}. An approximation algorithm with ratio $\left(1-\frac{1}{e}\right)^2$ has been proposed in~\cite{BPRSS16}. In the adaptive seeding with knapsack constraints problem, nodes in $X$ and in $N(X)$ are associated with a cost and the aim is to maximize the objective function while respecting a budget constraint. In this case, an $\left(1-\frac{1}{e}\right) \left(1-\frac{1}{e^{1-\frac{1}{e}}}\right)$-approximation algorithm is known~\cite{RSS15}. In the non-stochastic version of these problems, all the nodes in $N(X)$ become accessible with probability one. Even in this case it is not possible to approximate an optimal solution within a factor greater than $\left(1-\frac{1}{e^{1-\frac{1}{e}}}\right)$, unless $P=NP$.
A similar problem in which the core is made of the whole network and the network can be augmented by adding edges according to a given cost function has been shown to admit a $0.0878$-approximation algorithm~\cite{DSV17}. 
Finally, in~\cite{DSV16, TCS18} the authors consider the problem where the core is made of a give set of nodes and the network can be augmented by adding edges incident only to the nodes in the core. In the unit-cost version of the problem where the cost of adding any edge is constant and equal to $1$ the problem is $NP$-hard to be approximated within a constant factor greater than $1-(2e)^{-1}$. Then they provide a greedy approximation algorithm that guarantees an approximation factor of $1-\frac{1}{e}-\epsilon$, where
$\epsilon$ is any positive real number. 
Then, they study the more general problem where the cost of edges is in $[0,1]$ and propose an algorithm that
achieves an approximation guarantee of $1-\frac{1}{e}$ combining greedy and enumeration technique.

\section{Preliminaries} \label{preliminaries}
We are given a set $X$ of $n$ elements and a set $S$ of $m$ bins. Let us denote the cost of a bin $s\in S$ by $c(s)\in \mathbb{R}_{\geq 0}$. For each bin $s\in S$ and element $x\in X$, we denote by $c(s,x)$ the cost of associating $x$ to $s$.
Given a budget $k\in \mathbb{R}_{\geq 0}$, and a monotone submodular function $f: 2^X \rightarrow \mathbb{R}_{\geq 0}$\footnote{For a ground set $X$, a function $f:2^X\rightarrow \mathbb{R}_{\geq 0}$ is submodular if for any pair of sets $S\subseteq T \subseteq X$ and for any element $x\in X\setminus T$, $f(S\cup\{x\}) - f(S) \geq f(T\cup \{x\}) - f(T)$.}, our goal is to find a subset $X'$ of $X$ and a subset $S'\neq\emptyset$ of $S$ such that $c(S',X') = \sum_{s\in S'} c(s) + \sum_{x\in X'} \min_{s\in S'}c(s,x) \leq k$, and $f(X')$ is maximum.
We call this problem the Generalized Budgeted submodular set function Maximization problem (\probl).

Our problem generalizes several well-known problems. Indeed, by setting $c(s,x)=\infty$, we do not allow the association of element $x$ to bin $s$, while by setting $c(s,x)=0$ we allow to assign element $x$ to bin $s$ with no additional cost. Moreover, we relax the constraints related to the association of elements to bins by setting $c(s)=0$ for each $s\in S$, and $c(s_1,x)=c(s_2,x)$, for each $s_1,s_2\in S$ and $x\in X$. By suitably combining these conditions we can capture the following problems: budgeted maximum coverage problem~\cite{khuller1999budgeted}; non-stochastic adaptive seeding problem~\cite{BPRSS16} (also with knapsack constraints~\cite{RSS15}); monotone submodular set function subject to a knapsack constraint maximization~\cite{S04}. Moreover, our cost function is not submodular and thus that the setting considered in \cite{IB13} does not generalize our model.



Let us consider a partial solution $(S',X')$. 
Given a set $T\subseteq X\setminus X'$, we denote by $c_{\min}(T)$ the minimum cost of associating the elements in $T$ with a single bin in $S$, considering that the cost of bins in $S'$ has been already paid, formally:
\[
  c_{\min}(T)=\min_{s\in S}\left\{c_{S'}(s)+\sum_{x \in T}c(s,x)\right\},
\]
where $c_{S'}(s) = c(s)$ if $s\not\in S'$, and $c_{S'}(s) = 0$ if $s\in S'$. We call $c_{\min}(T)$ the \emph{marginal cost} of $T$ with respect to the partial solution $(S',X')$.
We define $s_{\min}(T)$ as the bin $s \in S$ needed to cover $T$ with cost $c_{\min}(T)$. Moreover, we denote by $\bar{c}(T)$ the cost of associating the elements in $T$ to $s_{\min}(T)$, $\bar{c}(T)= c_{\min}(T) - c_{S'}(s_{\min}(T))$.


The \emph{marginal increment} of $T\subseteq X$ with respect to the partial solution $(S',X')$ is defined as $f(X' \cup T)-f(X')$. To simplify the notation, we use $g(T)=f(X' \cup T)-f(X')$ to denote the marginal increment. 

In the algorithm in the next section, we will look for subsets of $X$ that maximize the ratio between the marginal increment and the marginal cost with respect to some partial solution. In the following we define a family of subsets of $X$ containing a set that approximates such maximal ratio. Given a partial solution $(S',X')$, we denote by $\cal F$ the family of subsets $T$ of $X$ that can be associated to bins in $S'\cup\{s\}$, for some single bin $s\in S$, with a cost such that $c(S'\cup\{s_{\min}(T)\},T)\leq k$, formally ${\cal F} =\left\{ T\in 2^{X\setminus X'}~|~ c(S'\cup\{s_{\min}({T})\},{T})\leq k  \right\} $.
A sub-family of $\cal F$ is an $\alpha$-\emph{list} with respect to $(S',X')$ if it contains a subset $T$ whose ratio between marginal increment and marginal cost is at least $\alpha$ times the optimal such ratio amongst all the subsets $\cal F$.
Formally, $L\subseteq {\cal F}$ is an $\alpha$-list with respect to $(S',X')$ if
%
\[
\max\left\{ \frac{g(T)}{c_{\min}(T)}~|~T\in L,c_{\min}(T)>0  \right\}  \geq  \alpha \cdot  \max\left\{ \frac{g(T)}{c_{\min}(T)}~|~T\in {\cal F},c_{\min}(T)>0  \right\}.
\]
Note that the sets that maximize the above formula are not necessarily singletons due to the bin opening cost. Moreover, the algorithm given in the next section build partial solutions $(S',X')$ in such a way that $c_{\min}(T)>0$, for each $T \in {\cal F}$.

\section{Greedy Algorithm}\label{sec:general}
In this section we give an algorithm that guarantees a $\frac{1}{2}\left(1-\frac{1}{e^\alpha}\right)$-approximation to the \probl problem,
 if we assume that we can compute, in polynomial time, an $\alpha$-list of polynomial size.  
 In the next sections we will give two algorithms to compute such lists for bounded values of $\alpha$.

The pseudo-code is reported in Algorithm~\ref{generalalgo}. In the first step (line \ref{generalalgo:zerocostbin}) we add all zero-cost bins to the solution.
Then, the algorithm finds two candidate solutions. The first one is found at lines~\ref{generalalgo:greedystart}--\ref{generalalgo:greedyend} with a greedy strategy as follows.
The algorithm iteratively constructs a partial solution $(S',X')$ by adding a subset $\hat{T}$ to $X'$ and a bin $s_{\min}(\hat{T})$ to $S'$. In particular, at each iteration, 
it first adds all the elements that can be associated to $S'$ with cost 0 (line~\ref{generalalgo:zerocost}). Then, it selects a subset $\hat{T}$ that maximizes the ratio between the marginal increment and the marginal cost amongst the elements of an $\alpha$-list $L$. Here, we assume that we have an algorithm to compute an $\alpha$-list $L$ w.r.t. $(S',X')$ (see line~\ref{generalalgo:alpha}). In the next sections, we will show how to compute $L$ in polynomial time for some bounded $\alpha$. The algorithm stops when adding the element with the maximum ratio would exceed the budget $k$ or when $X'=X$. Without loss of generality, we can assume that at each iteration, the sets in the $\alpha$-list $L$ do not contain any element in $X'$, since such elements do not increase the value of the marginal increment and possibly increase the marginal cost. This implies that at each iteration of the greedy procedure at least a new element in $X$ is added to $X'$ and then the number of iterations is $O(n)$.

Let $(S_G,X_G)$ be the first candidate solution computed at the end of the greedy procedure.
The second candidate solution (lines~\ref{generalalgo:max}--\ref{generalalgo:secondsol}) is computed by using the set $\hat{T}$ that is discarded in the last iteration of the greedy procedure because adding $\{s_{\min}(\hat{T})\}$ and $\hat{T}$ to $(S_G,X_G)$ would exceed the budget. Indeed, the second candidate solution is $(S_G\cup\{s_{\min}(\hat{T})\},\hat{T})$. Note that this solution is feasible because $\hat{T}$ is contained in the $\alpha$-list $L$ computed in the last iteration of the greedy algorithm. Therefore, by definition of $\alpha$-list, 
$c(S_G\cup\{s_{\min}(\hat{T})\},\hat{T})\leq k$.


The algorithm returns one of the two candidate solutions that maximizes the objective function.

The computational complexity of Algorithm~\ref{generalalgo} is $O(n\cdot (|L_{\max}| + cl))$, where $L_{\max}$ is the largest $\alpha$-list computed and $cl$ is the computational complexity of the algorithm at line~\ref{generalalgo:alpha}. In the next sections we will show that our algorithms construct the $\alpha$-lists in such a way that both $|L_{\max}|$ and $cl$ are polynomially bounded in the input size.

\begin{algorithm}[t]
    \SetKwInOut{Input}{Input}
    \SetKwInOut{Output}{Output}
    \Input{$S, X$ }
    \Output{$S',X'$}
    $S' := \emptyset$\;
    $X' := \emptyset$\;
    \lForEach{\rm $s\in S$ s.t. $c(s)=0$}{$S' := S'\cup \{s\}$\label{generalalgo:zerocostbin}}
    \Repeat{$c(S' \cup \{s_{\min}(\hat{T})\}, X'\cup \hat{T}) > k$ \rm or $X'= X$}{\label{generalalgo:greedystart}
      \lForEach{\rm $x\in X\setminus X'$ s.t. $c(s',x)=0$ and $s'\in S'$\label{generalalgo:outcondition}}{$X':=X'\cup \{x\}$\label{generalalgo:zerocost}}
      Build an $\alpha$-list $L$ w.r.t. $(S',X')$\;\label{generalalgo:alpha}
      $\hat{T}:=\arg\max_{T\in L}\frac{f(X' \cup T)-f(X')}{c_{\min}(T)}$\;\label{generalalgo:greedymax}
      \If{$c(S' \cup \{s_{\min}(\hat{T})\}, X'\cup \hat{T}) \leq k$ \label{generalalgo:budget}}{      
        $S':=S' \cup \{s_{\min}(\hat{T})\}$\;
        $X':=X' \cup \hat{T}$\;
      }
    }
    \label{generalalgo:greedyend}
    \If{$f(\hat{T})\geq f(X')$}{\label{generalalgo:max}
      $S':=S'\cup\{s_{\min}(\hat{T})\}$\;
      $X':=\hat{T}$\;\label{generalalgo:secondsol}
    }
    \Return $(S',X')$\;
    \caption{General Algorithm}
    \label{generalalgo}
\end{algorithm}

In what follows we analyze the approximation ratio of Algorithm~\ref{generalalgo}. The proof generalizes known arguments for monotone submodular maximization, see e.g.~\cite{CK08,khuller1999budgeted,S04}.

We give some additional definitions that will be used in the proof. We denote an optimal solution by $(S^*, X^*)$. Let us consider the iterations executed by the greedy algorithm. Let $l+1$ be the index of the iteration in which an element in the $\alpha$-list is not added to $X'$ because it violates the budget constraint\footnote{We can assume that this iteration exists, as otherwise the algorithm is able to select $X'=X$, which is the optimum.}.
For $i=1,2,\ldots,l$, we define $X_i'$ and $S_i'$ as the sets $X'$ and $S'$ at the end of the $i$-th iteration of the algorithm, respectively. Moreover, let $X_{l+1}' = X_l' \cup \{\hat{T}\}$ and $S_{l+1}' =  S_l' \cup \{s_{\min}(\hat{T})\}$, where $\hat{T}$ is the element selected at line~\ref{generalalgo:greedymax} of iteration $l+1$ (see Algorithm~\ref{generalalgo}). Let $c_i$ be the value of $c_{\min}(\hat{T})$ as computed at iteration $i$ of the greedy algorithm.
The next lemma will be used in the proof of Theorem~\ref{th:greedy}.
\begin{lemma}\label{lem:induction_min}
After each iteration $i=1,2,\dots,l+1$, 
\begin{equation*}
f(X_i')\geq \left( 1-\prod_{j=1}^i\left(1-\alpha\frac{c_j}{k}\right) \right)f(X^*).
\end{equation*}
\label{lemma2algogenerale}
\end{lemma}
The next  lemma will be used in the proof of the Lemma~\ref{lem:induction_min}.
\begin{lemma}\label{lemma1algogenerale}
After each iteration $i = 1, 2, \dots , l+1$, the following holds 
\[
f(X'_i)-f(X'_{i-1})\geq \frac{c_i}{k}\alpha (f(X^*)-f(X'_{i-1})). 
\]
\end{lemma}
\begin{proof}
Let us consider a partition of the elements in $X^*\setminus X_{i-1}'$  according to the bins they are assigned to in the optimal solution
$(S^*,X^*)$, that is the elements of each set $T_j$ in the partition are associated with the same bin in $(S^*,X^*)$ Formally, $X^*\setminus X_{i-1}' = T_1\cup T_2 \cup \ldots \cup T_\ell$ such that for each $j=1,2,\ldots,\ell$ and for each $x_1,x_2\in T_j$, $\arg\min_{s\in S^*}\{c(s,x_1)\} = \arg\min_{s\in S^*}\{c(s,x_2)\}$ and $T_j$ is maximal.

We first show the following:
\[
f(X^*)-f(X_{i-1}') \leq \sum_{j=1}^{\ell}(f(X_{i-1}' \cup T_j)-f(X_{i-1}')).
\]
Indeed, 
\begin{align*}
f(X^*)-f(X_{i-1}') \leq &f(X^*\cup X_{i-1}')-f(X_{i-1}') \\
=&\sum_{j=1}^\ell \left( f(X_{i-1}' \cup T_1\cup\ldots \cup T_j )- f(X_{i-1}'  \cup T_1\cup\ldots \cup T_{j-1} ) \right)\\
\leq&\sum_{j=1}^\ell(f(X_{i-1}' \cup T_j)-f(X_{i-1}')),
\end{align*}
where the last inequality follows by submodularity of $f$.

Let us denote by $c^*_i(T_j)$ the marginal cost of adding the elements $T_j$ to solution $(S'_{i-1},X'_{i-1})$, that is $c^*_i(T_j) = c(S_{i-1}'\cup\{s^*(T_j)\},X_{i-1}' \cup T_j)-c(S_{i-1}',X_{i-1}')$, where $s^*(T_j)$ is the bin in $S^*$ which all the elements of $T_j$ are associated with.
By definition of $\alpha$-list and the maximum at line~\ref{generalalgo:greedymax} it follows that, for each $j=1,2,\ldots,\ell$,
\[
\frac{f(X_{i}')-f(X_{i-1}')}{c_i} \geq \alpha \frac{f(X_{i-1}' \cup T_j)-f(X_{i-1}')}{c^*_i(T_j) }.
\]
Therefore,
\begin{align*}
\sum_{j=1}^\ell(f(X_{i-1}' \cup T_j)-f(X_{i-1}')) &\leq  \sum_{j=1}^\ell\frac{c^*_i(T_j)}{\alpha}\frac{f(X_{i}')-f(X_{i-1}')}{c_i} \\
&= \frac{f(X_{i}')-f(X_{i-1}')}{\alpha c_i} \sum_{j=1}^\ell c^*_i(T_j) \\
&\leq \frac{f(X_{i}')-f(X_{i-1}')}{\alpha c_i} k.
\end{align*}
It follows that:
\[
f(X^*)-f(X_{i-1}') \leq \frac{k}{c_i\alpha } \Bigl(f(X_i')-f(X_{i-1}')\Bigr),
\]
which implies the statement.
\end{proof}

\noindent {\bf Proof of Lemma~\ref{lem:induction_min}.}
For $i=1$ from Lemma~\ref{lemma1algogenerale} follows that $f(X'_1)\geq \alpha\frac{c_1}{k} f(X^*)$.
Applying Lemma \ref{lemma1algogenerale} and the inductive hypothesis we obtain:
\begin{align*}
f(X_i')&=f(X_{i-1}')+(f(X_i')-f(X_{i-1}')) \\
&\geq f(X_{i-1}')+\alpha\frac{c_i}{k}(f(X^*)-f(X_{i-1}'))\\
&= f(X_{i-1}')\left(1-\alpha\frac{c_i}{k}\right)+\alpha\frac{c_i}{k}f(X^*)\\
&\geq \left(1-\prod_{j=1}^{i-1}\left(1-\alpha\frac{c_j}{k}\right)\right)f(X^*)\left(1-\alpha\frac{c_i}{k}\right)+\alpha\frac{c_i}{k}f(X^*)\\
&= \left(1-\prod_{j=1}^{i}\left(1-\alpha\frac{c_j}{k}\right)\right)f(X^*).
\end{align*}\qed

Armed with Lemma~\ref{lem:induction_min}, we can prove Theorem~\ref{th:greedy}.
\begin{theorem}\label{th:greedy}
Algorithm \ref{generalalgo} guarantees an approximation factor of $\frac{1}{2}\left(1-\frac{1}{e^{\alpha}}\right)$ for \probl.
\end{theorem}
\begin{proof}
We observe that since $(S_{l+1}',X_{l+1}')$ violates the budget, then $c(S_{l+1}',X_{l+1}')> k$. Moreover, for a sequence of numbers $a_1,a_2,\ldots,a_n$ such that $\sum_{\ell=1}^n a_\ell = A$, the function $\left[ 1 - \prod_{i=1}^n \left(1 -\frac{a_i\cdot \alpha}{A}\right)\right]$ achieves its minimum when $a_i=\frac{A}{n}$ and that 
$
\left[ 1 - \prod_{i=1}^n \left(1 -\frac{a_i\cdot \alpha}{A}\right)\right]\geq 1-\left(1-\frac{\alpha}{n}\right)^n\geq 1-e^{-\alpha}.
$
Therefore, by applying Lemma~\ref{lem:induction_min} for $i=l+1$ and observing that $\sum_{\ell=1}^{l+1}c_\ell=c(S_{l+1}',X_{l+1}')$, we obtain:
\begin{align}\label{th:greedy_budget_line1}
 f(X_{l+1}') & \geq \left[ 1- \prod_{\ell=1}^{l+1}\left( 1 - \frac{c_\ell \cdot \alpha}{k}\right) \right]f(X^*)\\ \label{th:greedy_budget_line2}
                         & > \left[ 1- \prod_{\ell=1}^{l+1}\left( 1 - \frac{c_\ell \cdot \alpha}{c(S'_{l+1},X'_{l+1})}\right) \right]f(X^*)\\  \label{th:greedy_budget_line3}
                         & \geq \left[ 1- \left( 1 - \frac{\alpha}{(l+1)}\right)^{l+1} \right]f(X^*)
                          \geq \left(1-\frac{1}{e^{\alpha}}\right) f(X^*).
\end{align}
Since, by submodularity, $f(X_{l+1}')  \leq f(X_{l}') + f(\hat{T})$, where $\hat{T}$ is the set selected at iteration $l+1$, we get
\[
f(X_{l}') + f(\hat{T})\geq \left(1-\frac{1}{e^{\alpha}}\right) f(X^*).
\]
Hence, $\max\{ f(X'_l), f(\hat{T}) \}\geq \frac{1}{2}\left(1-\frac{1}{e^{\alpha}}\right) f(X^*)$. The theorem follows by observing that $\hat{T}$ is the set selected as the second candidate solution at lines~\ref{generalalgo:max}--\ref{generalalgo:secondsol} of Algorithm~\ref{generalalgo}.
\end{proof}

\section{Computing an $\alpha$-list for a particular case}\label{sec:epsilonblock}
In this section, we give a polynomial time algorithm to find a $(1-\epsilon)$-list with respect to a partial solution $(S',X')$ for the particular case in which, for a given parameter $\epsilon\in(0,1)$, the following condition holds:
\begin{equation}\label{eq:epsilonblock}
 \sum_{x\in T}c(s,x)\geq \frac{1}{\epsilon}c(s),
\end{equation}
for each $s \in S$ and for each $T\subseteq X$ such that $|T|=\frac{1}{\epsilon}$. We observe that this condition is fulfilled for any $\epsilon\in (0,1)$ in the case in which $c(s)=1$ and $c(s,x)\geq 1$, for each $s\in S$ and for each $x\in X$, which generalizes the non-stochastic adaptive seeding problem~\cite{BPRSS16}. Indeed, in this case $\sum_{x\in T}c(s,x)\geq |T|=\frac{1}{\epsilon}c(s)$, for each $s \in S$ and for each $T\subseteq X$, such that $|T|=\frac{1}{\epsilon}$.

We give a simple algorithm that returns a $(1-\epsilon)$-list with respect to a partial solution $(S',X')$. The algorithm works as follows: build a list which contains all the subsets $T$ of $X\setminus X'$ such that $|T|\leq \frac{1}{\epsilon}$ and $c(S'\cup\{s_{\min}(\hat{T})\},\hat{T})\leq k$.

Plugging this algorithm into line~\ref{generalalgo:alpha} of Algorithm~\ref{generalalgo}, we can guarantee an approximation factor of $\frac{1}{2}\left(1-\frac{1}{e}\right)-\epsilon'$, where $\epsilon'=\frac{1}{2e}\left(e^\epsilon-1\right)$ for \probl. 

We observe that the case in this section contains the problem of maximizing a submodular set function under a knapsack constraint as a special case. Indeed, it is enough to set $c(s)=0$, for each $s\in S$, and $c(s_1,x)=c(s_2,x)$, for each $s_1,s_2\in S$ and $x\in X$. Note that in this case Condition~\ref{eq:epsilonblock} is satisfied for any $\epsilon\in (0,1)$. A special case of submodular set function maximization is the maximum coverage problem, and since this latter is $NP$-hard to be approximated within a factor greater than $\left(1-\frac{1}{e}\right)$~\cite{F98}, then the gap between the
approximation factor of our algorithm and the best achievable one in polynomial time is $\frac{1}{2}$, unless $P=NP$.

It is easy to see that the computational complexity required by the algorithm in this section is $O(n^\frac{1}{\epsilon})$ and that $|L_{\max}|=O(n^\frac{1}{\epsilon})$.

In what follows, we assume that any set $T^*$ that maximizes the ratio between marginal increment and marginal cost has size greater than $\frac{1}{\epsilon}$, as otherwise the $\alpha$-list returned by our algorithm would contain such set.
The following two technical lemmata will be used in the analysis of the algorithm. 
\begin{lemma}\label{lem:difference}
Given a monotone submodular set function $f:2^X\rightarrow \mathbb{R}_{\geq 0}$, then, for any $X'\subseteq X$, the function $g(T)=f(X'\cup T)-f(X')$ is monotone and submodular.
\end{lemma}
\begin{proof}
It is easy to prove that $g$ is monotone.
We show  that for each pair of sets $T , S$ such that $T\subseteq  S\subseteq (X\setminus X')$ and for each $x\in X \setminus (X'\cup S)$, the increment in the value of $g$ that element $x$ causes in $S \cup \{x\}$ is smaller than the increment it produces in $T \cup \{x\}$, that is
\[
g(T \cup \{x\})-g(T) \geq g( S \cup \{x\})-g(S).
\]
By applying the definition we have:
\[
f(X'\cup T\cup \{x\})-f(X') -f(X'\cup T)+f(X') \geq f(X'\cup S \cup \{x\})-f(X') - f(X'\cup S)+f(X'),
\]
which is equivalent to:
\[
f(X'\cup T\cup \{x\})-f(X'\cup T)\geq f(X'\cup S \cup \{x\})- f(X'\cup S).
\]
The statement follows because $f$ is submodular.
\end{proof}

\begin{lemma}\label{lem:ratio}
Let us consider a monotone submodular set function $f:2^X\rightarrow \mathbb{R}_{\geq 0}$ and a cost function $c:2^X\rightarrow \mathbb{R}_{\geq 0}$ such that $c(T)=\sum_{x\in T}c(\{x\})$, for each $T\subseteq X$.
For each set $T\subseteq X$, if $T_y$ denotes the subset of $T$ such that $\frac{f(T_y)}{{c}(T_y)}$ is maximum and $|T_y|=y$, then $\frac{f(T)}{{c}(T)}\leq \frac{f(T_y)}{{c}(T_y)}$, for any $y\leq|T|$.
\end{lemma}
\begin{proof}
We show the following equivalent statement: given a set $T\subseteq X$, there exists a set $T'\subseteq T$, such that $|T'|=|T|-1$ and $\frac{f(T)}{{c}(T)}\leq \frac{f(T')}{{c}(T')}$.
Let $T=\{t_1,t_2,\ldots,t_\ell\}$, where ${c}(\{t_i\}) \leq {c}(\{t_j\})$, for each $i<j$, and let $\delta_i$ denote the marginal increment given by adding element $t_{i}$ to the set $\{t_1,t_2,\ldots,t_{i-1}\}$, that is, $\delta_i = f(\{t_1,t_2,\ldots,t_{i}\}) - f(\{t_1,t_2,\ldots,t_{i-1}\})$. We have that $f(T) = \sum_{i=1}^\ell \delta_i$ and ${c}(T) = \sum_{i=1}^\ell {c}(\{t_i\})$. We define $T'=T\setminus \{t_\ell\}$, i.e. we remove from $T$ an element with the maximum cost. We obtain $f(T') = \sum_{i=1}^{\ell-1} \delta_i$ and ${c}(T') = \sum_{i=1}^{\ell-1} {c}(\{t_i\})$. Since ${c}(\{t_\ell\})$ is maximum, then ${c}(\{t_\ell\})\geq \frac{\sum_{i=1}^{\ell-1} {c}(\{t_i\})}{\ell - 1} = \frac{{c}(T')}{\ell - 1}$. Moreover, by submodularity of $f$, $\delta_\ell \leq \frac{\sum_{i=1}^{\ell-1} \delta_i}{\ell - 1} =  \frac{f(T')}{\ell - 1}$. Therefore:
\[
 \qquad\qquad\qquad\frac{f(T)}{c(T)} = \frac{f(T') + \delta_\ell}{{c}(T')+ {c}(\{t_\ell\})} \leq \frac{f(T') + \frac{f(T')}{\ell - 1}}{{c}(T') + \frac{{c}(T')}{\ell - 1}} = \frac{f(T')}{c(T')} .
\]
\end{proof}

The next theorem shows the approximation ratio of the algorithm.
The main idea is to consider the subset $\hat{T}$ that maximizes the ratio between the marginal increment and marginal cost in $L$ and to derive a series of inequalities to lead us state that this value is greater than the ratio given by the optimal subset $T^*$ times the factor $(1-\epsilon)$. We first compare the ratio computed for $\hat{T}$ with that for $T^*_{\frac{1}{\epsilon}}$ that is a subset of cardinality $\frac{1}{\epsilon}$ of maximal ratio, then, by rewriting the marginal cost formula according to its definition and by exploiting Lemmata~\ref{lem:difference} and~\ref{lem:ratio}, and Condition~\eqref{eq:epsilonblock} we compare this ratio to that given by the subset $T^*$ and this last inequality concludes the theorem.  
\begin{theorem}
If for each $T\subseteq X$ such that $|T|=\frac{1}{\epsilon}$ and for each $s \in S$ we have $\sum_{x\in T}c(s,x)\geq \frac{1}{\epsilon}c(s)$, then the list $L$ made of all the subsets of $X\setminus X'$ of size at most $\frac{1}{\epsilon}$ and cost at most $k$ is a $(1-\epsilon)$-list.
\end{theorem}
\begin{proof}
We recall that $g(T)=f(X'\cup T) -f(X')$. Given a subset $T$ of $X\setminus X'$, we denote by $T_y$ a subset of $T$ such that $|T_y|= y$ and $\frac{f(T_y)}{\bar{c}(T_y)}$ is maximum.
Let $T^*$ be the subset of $X\setminus X'$ that maximizes the ratio between the marginal increment and the marginal cost.
Let $\hat{T}$ be the element of $L$ that maximizes $\frac{g(\hat{T})}{c_{\min}(\hat{T})}$. Since $|\hat{T}|\leq \frac{1}{\epsilon}$, then
\[
 \frac{g(\hat{T})}{c_{\min}(\hat{T})}\geq\frac{g\left(T_{\frac{1}{\epsilon}}^*\right)}{c_{\min}\left(T_{\frac{1}{\epsilon}}^*\right)} = \frac{g\left(T_{\frac{1}{\epsilon}}^*\right)}{c_{S'}\left(s_{\min}\left(T_{\frac{1}{\epsilon}}^*\right)\right)+\bar{c}\left(T_{\frac{1}{\epsilon}}^*\right)}.
\]
By the hypothesis of the theorem, $\bar{c}\left(T_{\frac{1}{\epsilon}}^*\right)\geq \frac{1}{\epsilon}c\left(s_{\min}\left(T_{\frac{1}{\epsilon}}^*\right)\right)$, moreover, $c\left(s_{\min}\left(T_{\frac{1}{\epsilon}}^*\right)\right) \geq c_{S'}\left(s_{\min}\left(T_{\frac{1}{\epsilon}}^*\right)\right)$ and then $c_{S'}\left(s_{\min}\left(T_{\frac{1}{\epsilon}}^*\right)\right)\leq \epsilon\bar{c}\left(T_{\frac{1}{\epsilon}}^*\right)$. Therefore,
\[
 \frac{g\left(T_{\frac{1}{\epsilon}}^*\right)}{c_{S'}\left(s_{\min}\left(T_{\frac{1}{\epsilon}}^*\right)\right)+\bar{c}\left(T_{\frac{1}{\epsilon}}^*\right)}\geq \frac{g\left(T_{\frac{1}{\epsilon}}^*\right)}{\epsilon \bar{c}\left(T_{\frac{1}{\epsilon}}^*\right)+\bar{c}\left(T_{\frac{1}{\epsilon}}^*\right)} = \frac{g\left(T_{\frac{1}{\epsilon}}^*\right)}{(\epsilon +1) \bar{c}\left(T_{\frac{1}{\epsilon}}^*\right)}.
\]
Since $f$ is monotone and submodular, then, by Lemma~\ref{lem:difference}, also $g\left(T_{\frac{1}{\epsilon}}^*\right)$ is submodular. By Lemma~\ref{lem:ratio} follows that 
\[
\frac{g\left(T_{\frac{1}{\epsilon}}^*\right)}{(\epsilon +1) \bar{c}\left(T_{\frac{1}{\epsilon}}^*\right)}\geq \frac{g(T^*)}{(\epsilon +1) \bar{c}(T^*)}.
\]
We now focus on the denominator, and we obtain that:
\begin{align*}
\frac{1}{(\epsilon +1) \bar{c}(T^*)} &=\frac{1+\epsilon -\epsilon}{(\epsilon +1) \bar{c}(T^*)}=
\frac{1}{\bar{c}(T^*)}-\frac{\epsilon}{(\epsilon + 1)\bar{c}(T^*)}\geq \\ 
& \frac{1}{\bar{c}(T^*)+c_{S'}(s_{\min}(T^*))}-\frac{\epsilon}{\epsilon \bar{c}(T^*)+\bar{c}(T^*)}.
\end{align*}
By applying the hypothesis $\bar{c}(T^*)\geq \frac{1}{\epsilon}c(s_{\min}(T^*))$, it follows that:
\begin{align*}
&\frac{1}{\bar{c}(T^*)+c_{S'}(s_{\min}(T^*))}-\frac{\epsilon}{\epsilon \bar{c}(T^*)+\bar{c}(T^*)}\geq\\
&\frac{1}{\bar{c}(T^*)+c_{S'}(s_{\min}(T^*))}-\frac{\epsilon}{c(s_{\min}(T^*))+\bar{c}(T^*)}\geq\\
&\frac{1}{\bar{c}(T^*)+c_{S'}(s_{\min}(T^*))}-\frac{\epsilon}{\bar{c}(T^*)+c_{S'}(s_{\min}(T^*))} = \frac{1-\epsilon}{c_{\min}(T^*)}.
\end{align*}
%
%
%
To conclude:
\[
\frac{g(\hat{T})}{c_{\min}(\hat{T})}\geq (1-\epsilon)\frac{g(T^*)}{c_{\min}(T^*)}.
\]
\end{proof}


\begin{section}{Computing an $\alpha$-list for the general case}\label{sec:slot}
In this section we give a polynomial time algorithm that builds a $\left(1-\frac{1}{e}\right)(1-\epsilon)$-list with respect to a partial solution $(S',X')$, for any $\epsilon\in (0,1)$. Using this algorithm as routine at line~\ref{generalalgo:alpha} of Algorithm~\ref{generalalgo}, we can guarantee an approximation factor of 
\[
\frac{1}{2}\left(1-\frac{1}{e^{\left(1-\frac{1}{e}\right)(1-\epsilon)}}\right)
\]
for \probl. We observe that this case generalizes the non-stochastic adaptive seeding with knapsack constraints problem, which cannot be approximated within a factor greater than $\left(1-\frac{1}{e^{1-\frac{1}{e}}}\right)$, unless $P=NP$~\cite{RSS15}. 
Then, the gap between the approximation factor of our algorithm and the best achievable one in polynomial time is $\frac{1}{2}$, unless $P=NP$.


In the algorithm of this section we make use of a procedure called \greedymaxcover to maximize the value of a monotone submodular function $g:2^X\rightarrow \mathbb{R}_{\geq 0}$, given a certain budget and costs associated to the elements of $X$. It is well-known that there exists a polynomial-time procedure that guarantees a $\left(1-\frac{1}{e}\right)$-approximation for this problem~\cite{S04}. 

Let us denote by  $\hat{c}$ the minimum possible positive value of functions $c(s)$ and $c(s,x)$, amongst all elements $x$ and bins $s$, i.e. $\hat{c} =\min\{ \min\{c(s) : s\in S,  c(s)>0\} , \min\{c(s,x) : s\in S, x\in X, c(s,x)>0\}\}$.

The main idea is to build an $\alpha$-list $L$ which contains approximate solutions to the problem of maximizing a monotone submodular set function subject to a knapsack constraint in which the budget increases by a factor $1+\epsilon$ starting from $\hat{c}$, and the cost of the elements are given by the cost of associating them to a single bin. In particular, we consider $q= \left\lfloor\log_{1+\epsilon} \left(\frac{k}{\hat{c}}\right)\right\rfloor + 1$ different budgets $B_i$ that iteratively increase by a factor $1+\epsilon$, i.e. $B_0=\hat{c}$ and $B_i = (1+\epsilon) B_{i-1}$, for $i=1,\ldots, q$. Moreover we define $B_{q+1}=k$. For each $i=0,\ldots, q+1$ and for each bin $s\in S$, we apply procedure \greedymaxcover with ground set $X$, budget $B_i$, and the cost of associating the elements to bin $s$ as cost function. Then, we add the set returned by \greedymaxcover to $L$.
In this way we consider a budget that is at most a factor $1+\epsilon$ greater than the cost of an optimal solution and the solution returned by \greedymaxcover for this budget has a value that is at most $1-\frac{1}{e}$ times  smaller than that of the optimal solution.

The pseudo-code of the algorithm is reported in Algorithm~\ref{ExpBudgetList}. The outer cycle at lines~\ref{ExpBudgetList:outerstart}--\ref{ExpBudgetList:outerend} iteratively selects a bin $s$ in $S$ and finds a list of sets of elements assigned to bin $s$. The inner cycle at lines~\ref{ExpBudgetList:innerstart}--\ref{ExpBudgetList:innerend}, at each iteration $i$, calls procedure \greedymaxcover which uses $g$ as function to maximize, $\hat{c}(1+\epsilon)^i -c_{S'}(s)$ as budget and the cost of associating the elements to bin $s$ as cost function (to compute this costs, we only pass $s$ as a parameter to \greedymaxcover).
The budget is increased by a factor $(1+\epsilon)$ until $\hat{c}(1+\epsilon)^i \geq k$. Finally the algorithm runs \greedymaxcover with the full budget $k$. See Figure \ref{ExpImg} for an illustration. 

\begin{figure}[t]
\scalebox{0.7}{
\ifx\du\undefined
  \newlength{\du}
\fi
\setlength{\du}{15\unitlength}
\begin{tikzpicture}
\pgftransformxscale{1.000000}
\pgftransformyscale{-1.000000}
\definecolor{dialinecolor}{rgb}{0.000000, 0.000000, 0.000000}
\pgfsetstrokecolor{dialinecolor}
\definecolor{dialinecolor}{rgb}{1.000000, 1.000000, 1.000000}
\pgfsetfillcolor{dialinecolor}
\pgfsetlinewidth{0.100000\du}
\pgfsetdash{}{0pt}
\pgfsetdash{}{0pt}
\pgfsetbuttcap
{
\definecolor{dialinecolor}{rgb}{0.000000, 0.000000, 0.000000}
\pgfsetfillcolor{dialinecolor}
\definecolor{dialinecolor}{rgb}{0.000000, 0.000000, 0.000000}
\pgfsetstrokecolor{dialinecolor}
\draw (7.934493\du,9.004594\du)--(37.634493\du,9.004594\du);
}
\pgfsetlinewidth{0.100000\du}
\pgfsetdash{}{0pt}
\pgfsetdash{}{0pt}
\pgfsetbuttcap
{
\definecolor{dialinecolor}{rgb}{0.000000, 0.000000, 0.000000}
\pgfsetfillcolor{dialinecolor}
\definecolor{dialinecolor}{rgb}{0.000000, 0.000000, 0.000000}
\pgfsetstrokecolor{dialinecolor}
\draw (7.955204\du,8.987726\du)--(7.955204\du,8.351330\du);
}
\pgfsetlinewidth{0.100000\du}
\pgfsetdash{}{0pt}
\pgfsetdash{}{0pt}
\pgfsetbuttcap
{
\definecolor{dialinecolor}{rgb}{0.000000, 0.000000, 0.000000}
\pgfsetfillcolor{dialinecolor}
\definecolor{dialinecolor}{rgb}{0.000000, 0.000000, 0.000000}
\pgfsetstrokecolor{dialinecolor}
\draw (9.380527\du,8.970551\du)--(9.380527\du,8.334155\du);
}
\pgfsetlinewidth{0.100000\du}
\pgfsetdash{}{0pt}
\pgfsetdash{}{0pt}
\pgfsetbuttcap
{
\definecolor{dialinecolor}{rgb}{0.000000, 0.000000, 0.000000}
\pgfsetfillcolor{dialinecolor}
\definecolor{dialinecolor}{rgb}{0.000000, 0.000000, 0.000000}
\pgfsetstrokecolor{dialinecolor}
\draw (11.462956\du,8.938731\du)--(11.462956\du,8.302335\du);
}
\pgfsetlinewidth{0.100000\du}
\pgfsetdash{}{0pt}
\pgfsetdash{}{0pt}
\pgfsetbuttcap
{
\definecolor{dialinecolor}{rgb}{0.000000, 0.000000, 0.000000}
\pgfsetfillcolor{dialinecolor}
\definecolor{dialinecolor}{rgb}{0.000000, 0.000000, 0.000000}
\pgfsetstrokecolor{dialinecolor}
\draw (14.641401\du,8.977622\du)--(14.641401\du,8.341226\du);
}
\pgfsetlinewidth{0.100000\du}
\pgfsetdash{}{0pt}
\pgfsetdash{}{0pt}
\pgfsetbuttcap
{
\definecolor{dialinecolor}{rgb}{0.000000, 0.000000, 0.000000}
\pgfsetfillcolor{dialinecolor}
\definecolor{dialinecolor}{rgb}{0.000000, 0.000000, 0.000000}
\pgfsetstrokecolor{dialinecolor}
\draw (19.729034\du,8.945802\du)--(19.729034\du,8.309406\du);
}
\pgfsetlinewidth{0.100000\du}
\pgfsetdash{}{0pt}
\pgfsetdash{}{0pt}
\pgfsetbuttcap
{
\definecolor{dialinecolor}{rgb}{0.000000, 0.000000, 0.000000}
\pgfsetfillcolor{dialinecolor}
\definecolor{dialinecolor}{rgb}{0.000000, 0.000000, 0.000000}
\pgfsetstrokecolor{dialinecolor}
\draw (37.579945\du,8.984693\du)--(37.579945\du,8.348297\du);
}
\pgfsetlinewidth{0.100000\du}
\pgfsetdash{}{0pt}
\pgfsetdash{}{0pt}
\pgfsetbuttcap
{
\definecolor{dialinecolor}{rgb}{0.000000, 0.000000, 0.000000}
\pgfsetfillcolor{dialinecolor}
\definecolor{dialinecolor}{rgb}{0.000000, 0.000000, 0.000000}
\pgfsetstrokecolor{dialinecolor}
\draw (35.278312\du,8.988229\du)--(35.278312\du,8.351833\du);
}
\definecolor{dialinecolor}{rgb}{0.000000, 0.000000, 0.000000}
\pgfsetstrokecolor{dialinecolor}
\node[anchor=west] at (23.794396\du,8.351330\du){};
\definecolor{dialinecolor}{rgb}{0.000000, 0.000000, 0.000000}
\pgfsetstrokecolor{dialinecolor}
\node[anchor=west] at (26.269269\du,8.457396\du){.   .   .};
\definecolor{dialinecolor}{rgb}{0.000000, 0.000000, 0.000000}
\pgfsetstrokecolor{dialinecolor}
\node[anchor=west] at (7.495584\du,9.977676\du){$0$};
\definecolor{dialinecolor}{rgb}{0.000000, 0.000000, 0.000000}
\pgfsetstrokecolor{dialinecolor}
\node[anchor=west] at (37.084468\du,9.949460\du){$k$};
\definecolor{dialinecolor}{rgb}{0.000000, 0.000000, 0.000000}
\pgfsetstrokecolor{dialinecolor}
\node at (9.404773\du,7.927066\du){$\hat{c}$};
\definecolor{dialinecolor}{rgb}{0.000000, 0.000000, 0.000000}
\pgfsetstrokecolor{dialinecolor}
\node at (11.526093\du,7.891711\du){$\hat{c}(1+\epsilon)$};
\definecolor{dialinecolor}{rgb}{0.000000, 0.000000, 0.000000}
\pgfsetstrokecolor{dialinecolor}
\node at (35.284881\du,7.962422\du){$\hat{c}(1+\epsilon)^i$};
\definecolor{dialinecolor}{rgb}{0.000000, 0.000000, 0.000000}
\pgfsetstrokecolor{dialinecolor}
\node at (14.637363\du,7.891711\du){$\hat{c}(1+\epsilon)^2$};
\definecolor{dialinecolor}{rgb}{0.000000, 0.000000, 0.000000}
\pgfsetstrokecolor{dialinecolor}
\node at (19.728532\du,7.927066\du){$\hat{c}(1+\epsilon)^3$};
\end{tikzpicture}
}
\centering
\caption{Growth of the budget $B_i$ in the inner cycle of the algorithm.}
\label{ExpImg}
\end{figure}
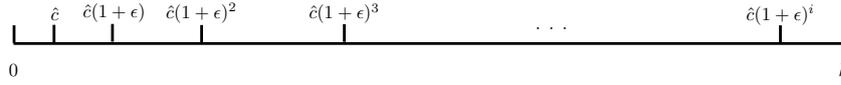


%

\begin{algorithm}[t]
    \SetKwInOut{Input}{Input}
    \SetKwInOut{Output}{Output} 
    \Input{$S, X,S',X', k, \epsilon$}
    \Output{$L$}
    $L := \emptyset$\;
    \ForEach{$s\in S$}{\label{ExpBudgetList:outerstart}
		$i := 0$\;
    		\While{$\hat{c}(1+\epsilon)^i<k$}{\label{ExpBudgetList:innerstart}
    				$B_i := \hat{c}(1+\epsilon)^i$\;
    				$T_i(s) := \greedymaxcover(X,s,B_i-c_{S'}(s))$\;
				$L := L\cup \{T_i(s)\}$\;    				
    				$i := i+1$
    		}\label{ExpBudgetList:innerend}
    $B_i := k$\;
    $T_i(s) := \greedymaxcover(X,s,B_i-c_{S'}(s))$\;
    $L := L\cup \{T_i(s)\}$\;
    }\label{ExpBudgetList:outerend}
    \Return $L$\;
    \caption{Exponential Budget Greedy}
    \label{ExpBudgetList}
\end{algorithm}

We call $q$ the value of $i$ at the end of the last iteration in the inner cycle of the algorithm. Let $T_j$ be the set in $L$ that maximizes the ratio between $g(T_j)$ and its assigned budget, that is:
\begin{equation}\label{eq:maxtj}
T_j=\arg\max\left\{\frac{g(T_i(s))}{B_i}: s\in S, i=0,1,\ldots,q+1\right\}.
\end{equation}
In order to bound the approximation ratio, we consider $\bar{X}^*$ as the set with the optimal ratio $\frac{g(\bar{X}^*)}{c_{\min}(\bar{X}^*)}$ amongst any possible subset of items. Let $B_l$ be the smallest value of $B_i$, for $i\in \{0,1,\dots,q+1\}$, that is greater than or equal to the cost of an optimal solution, that is the smallest $B_l$ such that $B_l \geq c_{\min}(\bar{X}^*)$. See figure \ref{OptImg} for an illustration.
We call $T^*_l$ the set in $L$ that has the highest ratio $\frac{g(T_l)}{B_l}$ amongst those computed by \greedymaxcover with budget $B_l$, i.e. $T^*_l = \max\left\{ \frac{g(T_l(s))}{B_l}: s\in S \right\}$.
We also denote the set that maximizes $g(X_l^*)$ with budget $B_l$ by $X_l^*$.

The idea of the approximation analysis is that an optimal solution $\bar{X}^*$ has a value of $g$ that is at most $g(X^*_l)$ and a cost that is at most $1+\epsilon$ times smaller than $B_l$, while the number of iterations remains polynomial since the size of the intervals grows exponentially. The next two theorems show the bounds on approximation ratio, computational complexity and size of $L$.

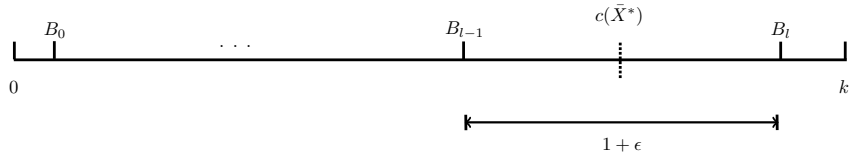
\begin{figure}[t]
\scalebox{0.7}{
\ifx\du\undefined
  \newlength{\du}
\fi
\setlength{\du}{15\unitlength}
\begin{tikzpicture}
\pgftransformxscale{1.000000}
\pgftransformyscale{-1.000000}
\definecolor{dialinecolor}{rgb}{0.000000, 0.000000, 0.000000}
\pgfsetstrokecolor{dialinecolor}
\definecolor{dialinecolor}{rgb}{1.000000, 1.000000, 1.000000}
\pgfsetfillcolor{dialinecolor}
\pgfsetlinewidth{0.100000\du}
\pgfsetdash{}{0pt}
\pgfsetdash{}{0pt}
\pgfsetbuttcap
{
\definecolor{dialinecolor}{rgb}{0.000000, 0.000000, 0.000000}
\pgfsetfillcolor{dialinecolor}
\definecolor{dialinecolor}{rgb}{0.000000, 0.000000, 0.000000}
\pgfsetstrokecolor{dialinecolor}
\draw (7.934493\du,9.004594\du)--(37.634493\du,9.004594\du);
}
\pgfsetlinewidth{0.100000\du}
\pgfsetdash{}{0pt}
\pgfsetdash{}{0pt}
\pgfsetbuttcap
{
\definecolor{dialinecolor}{rgb}{0.000000, 0.000000, 0.000000}
\pgfsetfillcolor{dialinecolor}
\definecolor{dialinecolor}{rgb}{0.000000, 0.000000, 0.000000}
\pgfsetstrokecolor{dialinecolor}
\draw (7.955204\du,8.987726\du)--(7.955204\du,8.351330\du);
}
\pgfsetlinewidth{0.100000\du}
\pgfsetdash{}{0pt}
\pgfsetdash{}{0pt}
\pgfsetbuttcap
{
\definecolor{dialinecolor}{rgb}{0.000000, 0.000000, 0.000000}
\pgfsetfillcolor{dialinecolor}
\definecolor{dialinecolor}{rgb}{0.000000, 0.000000, 0.000000}
\pgfsetstrokecolor{dialinecolor}
\draw (9.380527\du,8.970551\du)--(9.380527\du,8.334155\du);
}
\pgfsetlinewidth{0.100000\du}
\pgfsetdash{}{0pt}
\pgfsetdash{}{0pt}
\pgfsetbuttcap
{
\definecolor{dialinecolor}{rgb}{0.000000, 0.000000, 0.000000}
\pgfsetfillcolor{dialinecolor}
\definecolor{dialinecolor}{rgb}{0.000000, 0.000000, 0.000000}
\pgfsetstrokecolor{dialinecolor}
\draw (23.971675\du,8.981158\du)--(23.971675\du,8.344762\du);
}
\pgfsetlinewidth{0.100000\du}
\pgfsetdash{}{0pt}
\pgfsetdash{}{0pt}
\pgfsetbuttcap
{
\definecolor{dialinecolor}{rgb}{0.000000, 0.000000, 0.000000}
\pgfsetfillcolor{dialinecolor}
\definecolor{dialinecolor}{rgb}{0.000000, 0.000000, 0.000000}
\pgfsetstrokecolor{dialinecolor}
\draw (37.579945\du,8.984693\du)--(37.579945\du,8.348297\du);
}
\pgfsetlinewidth{0.100000\du}
\pgfsetdash{}{0pt}
\pgfsetdash{}{0pt}
\pgfsetbuttcap
{
\definecolor{dialinecolor}{rgb}{0.000000, 0.000000, 0.000000}
\pgfsetfillcolor{dialinecolor}
\definecolor{dialinecolor}{rgb}{0.000000, 0.000000, 0.000000}
\pgfsetstrokecolor{dialinecolor}
\draw (35.278312\du,8.988229\du)--(35.278312\du,8.351833\du);
}
\definecolor{dialinecolor}{rgb}{0.000000, 0.000000, 0.000000}
\pgfsetstrokecolor{dialinecolor}
\node[anchor=west] at (23.794396\du,8.351330\du){};
\definecolor{dialinecolor}{rgb}{0.000000, 0.000000, 0.000000}
\pgfsetstrokecolor{dialinecolor}
\node[anchor=west] at (14.990916\du,8.492752\du){.   .   .};
\definecolor{dialinecolor}{rgb}{0.000000, 0.000000, 0.000000}
\pgfsetstrokecolor{dialinecolor}
\node[anchor=west] at (7.495584\du,9.977676\du){$0$};
\definecolor{dialinecolor}{rgb}{0.000000, 0.000000, 0.000000}
\pgfsetstrokecolor{dialinecolor}
\node[anchor=west] at (37.084468\du,9.949460\du){$k$};
\definecolor{dialinecolor}{rgb}{0.000000, 0.000000, 0.000000}
\pgfsetstrokecolor{dialinecolor}
\node at (9.404773\du,7.927066\du){$B_0$};
\definecolor{dialinecolor}{rgb}{0.000000, 0.000000, 0.000000}
\pgfsetstrokecolor{dialinecolor}
\node at (24.006528\du,7.927066\du){$B_{l-1}$};
\definecolor{dialinecolor}{rgb}{0.000000, 0.000000, 0.000000}
\pgfsetstrokecolor{dialinecolor}
\node at (35.284881\du,7.927066\du){$B_l$};
\pgfsetlinewidth{0.100000\du}
\pgfsetdash{{1.000000\du}{1.000000\du}}{0\du}
\pgfsetdash{{1.000000\du}{1.000000\du}}{0\du}
\pgfsetbuttcap
{
\definecolor{dialinecolor}{rgb}{0.000000, 0.000000, 0.000000}
\pgfsetfillcolor{dialinecolor}
}
\definecolor{dialinecolor}{rgb}{0.000000, 0.000000, 0.000000}
\pgfsetstrokecolor{dialinecolor}
\draw (29.557316\du,8.886686\du)--(29.557316\du,9.159478\du);
\pgfsetlinewidth{0.100000\du}
\pgfsetdash{}{0pt}
\definecolor{dialinecolor}{rgb}{0.000000, 0.000000, 0.000000}
\pgfsetstrokecolor{dialinecolor}
\draw (29.557316\du,8.386686\du)--(29.557316\du,8.486686\du);
\definecolor{dialinecolor}{rgb}{0.000000, 0.000000, 0.000000}
\pgfsetstrokecolor{dialinecolor}
\draw (29.557316\du,8.553352\du)--(29.557316\du,8.653352\du);
\definecolor{dialinecolor}{rgb}{0.000000, 0.000000, 0.000000}
\pgfsetstrokecolor{dialinecolor}
\draw (29.557316\du,8.720019\du)--(29.557316\du,8.820019\du);
\pgfsetlinewidth{0.100000\du}
\pgfsetdash{}{0pt}
\definecolor{dialinecolor}{rgb}{0.000000, 0.000000, 0.000000}
\pgfsetstrokecolor{dialinecolor}
\draw (29.557316\du,9.659478\du)--(29.557316\du,9.559478\du);
\definecolor{dialinecolor}{rgb}{0.000000, 0.000000, 0.000000}
\pgfsetstrokecolor{dialinecolor}
\draw (29.557316\du,9.492811\du)--(29.557316\du,9.392811\du);
\definecolor{dialinecolor}{rgb}{0.000000, 0.000000, 0.000000}
\pgfsetstrokecolor{dialinecolor}
\draw (29.557316\du,9.326144\du)--(29.557316\du,9.226144\du);
\definecolor{dialinecolor}{rgb}{0.000000, 0.000000, 0.000000}
\pgfsetstrokecolor{dialinecolor}
\node at (29.521961\du,7.396736\du){$c(\bar{X}^*)$};
\pgfsetlinewidth{0.080000\du}
\pgfsetdash{}{0pt}
\pgfsetdash{}{0pt}
\pgfsetbuttcap
{
\definecolor{dialinecolor}{rgb}{0.000000, 0.000000, 0.000000}
\pgfsetfillcolor{dialinecolor}
\pgfsetarrowsstart{to}
\pgfsetarrowsend{to}
\definecolor{dialinecolor}{rgb}{0.000000, 0.000000, 0.000000}
\pgfsetstrokecolor{dialinecolor}
\draw (24.015367\du,11.250468\du)--(35.205331\du,11.250468\du);
}
\definecolor{dialinecolor}{rgb}{0.000000, 0.000000, 0.000000}
\pgfsetstrokecolor{dialinecolor}
\node at (29.619188\du,12.063641\du){$1+\epsilon$};
\pgfsetlinewidth{0.100000\du}
\pgfsetdash{}{0pt}
\pgfsetdash{}{0pt}
\pgfsetbuttcap
{
\definecolor{dialinecolor}{rgb}{0.000000, 0.000000, 0.000000}
\pgfsetfillcolor{dialinecolor}
\definecolor{dialinecolor}{rgb}{0.000000, 0.000000, 0.000000}
\pgfsetstrokecolor{dialinecolor}
\draw (24.049508\du,10.977629\du)--(24.055758\du,11.533879\du);
}
\pgfsetlinewidth{0.100000\du}
\pgfsetdash{}{0pt}
\pgfsetdash{}{0pt}
\pgfsetbuttcap
{
\definecolor{dialinecolor}{rgb}{0.000000, 0.000000, 0.000000}
\pgfsetfillcolor{dialinecolor}
\definecolor{dialinecolor}{rgb}{0.000000, 0.000000, 0.000000}
\pgfsetstrokecolor{dialinecolor}
\draw (35.148459\du,10.972562\du)--(35.154709\du,11.528812\du);
}
\end{tikzpicture}
}
\centering
\caption{Notation used in Theorem \ref{ApproxExpBudgetList}.}
\label{OptImg}
\end{figure}

\begin{theorem}\label{ApproxExpBudgetList}
The list $L$ built by Algorithm \ref{ExpBudgetList}, is a  $\left(1-\frac{1}{e}\right)(1-\epsilon)$-list.
\end{theorem}
\begin{proof}
Since, by construction, $c_{\min}(T_j)\leq B_j$, and, by Equation~\ref{eq:maxtj}, $\frac{g(T_j)}{B_j}$ is maximum, then
\[
\frac{g(T_j)}{c_{\min}(T_j)}\geq\frac{g(T_j)}{B_j}\geq \frac{g(T^*_l)}{B_l}.
\]
Procedure \greedymaxcover guarantees a $\left(1-\frac{1}{e}\right)$-approximation, then 
\[
\frac{g(T^*_l)}{B_l}\geq \left(1-\frac{1}{e}\right)\frac{g(X_l^*)}{B_l}.
\]
Moreover, since function $g$ is monotone and $c_{\min}(\bar{X}^*)\leq B_l$, then $g(X_l^*)\geq g(\bar{X}^*)$, and therefore:
\[
\left(1-\frac{1}{e}\right)\frac{g(X_l^*)}{B_l}\geq \left(1-\frac{1}{e}\right)\frac{g(\bar{X}^*)}{B_l}.
\]
We defined $B_l$ as the smallest value of $B_i$ that is at least $c_{\min}(\bar{X}^*)$, this implies that $B_{l-1}\leq c_{\min}(\bar{X}^*)$. Moreover the ratio between $B_l$ and $B_{l-1}$ is $1+\epsilon$.
It follows that $B_l\leq (1+\epsilon)c_{\min}(\bar{X}^*)$, which implies:
\[
\left(1-\frac{1}{e}\right)\frac{g(\bar{X}^*)}{B_l}\geq \left(1-\frac{1}{e}\right)\left(\frac{1}{1+\epsilon}\right)\frac{g(\bar{X}^*)}{c_{\min}(\bar{X}^*)}\geq \left(1-\frac{1}{e}\right)\left(1-\epsilon\right)\frac{g(\bar{X}^*)}{c_{\min}(\bar{X}^*)}
\]
The last inequality holds since $\frac{1}{1+\epsilon}=1-\frac{\epsilon}{1+\epsilon}\geq 1-\epsilon$, for any $\epsilon>0$, and this concludes the proof.
\end{proof}

\begin{theorem}\label{ExpBudgetList:compl}
Algorithm \ref{ExpBudgetList} requires $O\left(\frac{1}{\epsilon} m\cdot {\text gr(n)}\cdot \log\frac{k}{\hat{c}}\right)$ computational time, where $\text gr(n)$ is the computational time of \greedymaxcover, and $|L_{\max}| = O\left(\frac{1}{\epsilon} m \log\frac{k}{\hat{c}}\right)$.
\end{theorem}
\begin{proof}
The outer for cycle requires $m$ iterations.
We now bound the number  $q$ of iteration of the inner cycle of the algorithm.
By the exit condition of the cycle, we have:
$\hat{c}\cdot (1+\epsilon)^q < k$,
which is equivalent to:
$q< \log_{1+\epsilon}\left( \frac{k}{\hat{c}} \right)$.
Since for $\epsilon<1$, $\log_{1+\epsilon}\left( \frac{k}{\hat{c}} \right) =O\left( \frac{1}{\epsilon}\log\frac{k}{\hat{c}} \right)$, the statement follows.
\end{proof}
We observe that $O(\log\frac{k}{\hat{c}})$ is polynomially bounded in the size of the input.


\end{section}

\section{Bi-criterion approximation algorithm}\label{sec:bicriteria}
In this section we extend the results given in Section~\ref{sec:general} providing a bi-criterion approximation algorithm that guarantees a $\frac{1}{2}\left(1-\frac{1}{e^{\alpha\beta}}\right)$-approximation to the \probl problem, if we allow an extra budget up to a factor $\beta\geq 1$. We notice that, if $\beta=1$, i.e. we do not increase the budget, the approximation factor is $\frac{1}{2}\left(1-\frac{1}{e^\alpha}\right)$, while if $\beta=\frac{1}{\alpha}\ln \left(\frac{1}{2\epsilon}\right)$ the algorithm achieves an approximation factor of $\frac{1}{2}-\epsilon$, for any arbitrarily small $\epsilon> 0$.

The algorithm is slightly different from Algorithm~\ref{generalalgo} and it is reported in Algorithm~\ref{algo_bicriteria}. In this algorithm, we allow to exceed the given budget $k$ by a factor $\beta$. In particular we modify lines \ref{algo:if} and \ref{algo:beta}, admitting a greater budget respect to Algorithm \ref{generalalgo}.

\begin{algorithm}[t]
    \SetKwInOut{Input}{Input}
    \SetKwInOut{Output}{Output}
    \Input{$S, X$ }
    \Output{$S',X'$}
    $S' := \emptyset$\;
    $X' := \emptyset$\;
    \lForEach{\rm $s\in S$ s.t. $c(s)=0$}{$S' := S'\cup \{s\}$\label{algo:zerocostbin}}
    \Repeat{$c(S' \cup \{s_{\min}(\hat{T})\}, X'\cup \hat{T}) > \beta k$ \label{algo:beta} \rm or $X'= X$}{\label{algo:greedystart}
      \lForEach{\rm $x\in X\setminus X'$ s.t. $c(s',x)=0$ and $s'\in S'$}{$X':=X'\cup \{x\}$\label{algo:zerocost}}
      Build an $\alpha$-list $L$ w.r.t. $(S',X')$\;\label{algo:alpha}
      $\hat{T}:=\arg\max_{T\in L}\frac{f(X' \cup T)-f(X')}{c_{\min}(T)}$\;\label{algo:greedymax}
      \If{$c(S' \cup \{s_{\min}(\hat{T})\}, X'\cup \hat{T}) \leq \beta k$\label{algo:if}  }{    
        $S':=S' \cup \{s_{\min}(\hat{T})\}$\;
        $X':=X' \cup \hat{T}$\;
      }
    }
    \label{algo:greedyend}
    \If{$f(\hat{T})\geq f(X')$}{
      $S':=S'\cup\{s_{\min}(\hat{T})\}$\;
      $X':=\hat{T}$\;\label{algo:secondsol}
    }
    \Return $(S',X')$\;
    \caption{Bi-criterion Algorithm}
    \label{algo_bicriteria}
\end{algorithm}

In the next theorem we show the approximation ratio of this algorithm.

\begin{theorem}\label{th:greedy-bicriteria}
There exists an algorithm that guarantees a $\left[\frac{1}{2}\left(1-\frac{1}{e^{\alpha\beta}}\right),\beta\right]$ bi-criterion approximation for \probl, for any $\beta\geq 1$.
\end{theorem}
\begin{proof}
We observe that since $(S_{l+1}',X_{l+1}')$ violates the budget, then $c(S_{l+1}',X_{l+1}')>\beta k$. Moreover, for a sequence of numbers $a_1,a_2,\ldots,a_n$ such that $\sum_{\ell=1}^n a_\ell = A$, the function $\left[ 1 - \prod_{i=1}^n \left(1 -\frac{a_i\cdot \alpha}{A}\right)\right]$ achieves its minimum when $a_i=\frac{A}{n}$ and that 
$
\left[ 1 - \prod_{i=1}^n \left(1 -\frac{a_i\cdot \alpha}{A}\right)\right]\geq 1-\left(1-\frac{\alpha}{n}\right)^n\geq 1-e^{-\alpha}.
$
Therefore, by applying Lemma~\ref{lem:induction_min} for $i=l+1$ and observing that $\sum_{\ell=1}^{l+1}c_\ell=c(S_{l+1}',X_{l+1}')$, we obtain:
\begin{align}
 f(X_{l+1}') & \geq \left[ 1- \prod_{\ell=1}^{l+1}\left( 1 - \frac{c_\ell \cdot \alpha}{k}\right) \right]f(X^*)\\ \label{th:greedy-bicriteria-appendix_line2}
                         & > \left[ 1- \prod_{\ell=1}^{l+1}\left( 1 - \frac{c_\ell \cdot \alpha \cdot \beta}{c(S'_{l+1},X'_{l+1})}\right) \right]f(X^*)\\ 
                         & \geq \left[ 1- \left( 1 - \frac{\alpha  \cdot \beta}{(l+1)}\right)^{l+1} \right]f(X^*)
                          \geq \left(1-\frac{1}{e^{\alpha\beta}}\right) f(X^*).
\end{align}
Since, by submodularity, $f(X_{l+1}')  \leq f(X_{l}') + f(T)$, where $T$ is the set selected at iteration $l+1$, we get
\[
f(X_{l}') + f(T)\geq \left(1-\frac{1}{e^{\alpha\beta}}\right) f(X^*).
\]
Hence, $\max\{ f(X_l), f(T) \}\geq \frac{1}{2}\left(1-\frac{1}{e^{\alpha\beta}}\right) f(X^*)$. The theorem follows by observing that $T$ is the set selected as the second candidate solution at lines~\ref{generalalgo:max}--\ref{generalalgo:secondsol} of Algorithm~\ref{algo_bicriteria}.
\end{proof}



\section{Conclusion}
In this paper we defined a new challenging problem which leads to many open problems and new research questions, we referred to it as the generalized budgeted submodular set function maximization problem.

The main open problem is to close the gap between the known hardness result of $1-\frac{1}{e^\alpha}$, where $\alpha=1$ for the MC problem~\cite{F98} and $\alpha=1-\frac{1}{e}$ for the non-stochastic adaptive seeding problem with knapsack constraint problem~\cite{RSS15}, and our approximation bound of $\frac{1}{2}\left(1-\frac{1}{e^\alpha}\right)$.
One possibility to get rid of the $\frac{1}{2}$ factor could be to use the partial enumeration technique exploited in specific subproblems (e.g.  budgeted maximum coverage problem~\cite{khuller1999budgeted} and monotone submodular set function subject to a knapsack constraint maximization problem~\cite{S04}). However, this requires that each greedy step selects a single element of $X$, to be added to a partial solution $X'$, while our greedy algorithm selects a subset of $X\setminus X'$ that maximizes the ratio between its marginal increment in the objective function and its marginal cost. Note that this set can contain more than one element in order to ensure that the ratio is non-increasing at each iteration of the greedy algorithm, which is needed to apply the analysis in~\cite{khuller1999budgeted} and~\cite{S04}.


Other research directions, that deserve further investigation, include the study of the \probl considering different cost functions and also different objective functions where the profit given by an element $x$ depends on the bin $s$ which it is associated with.
It would be interesting also to analyse \probl in the case that each bin $s\in S$ has its own budget $k$ to use in order to maximize the objective function.   
\bibliographystyle{plainurl}
\bibliography{Bibliography} 

\end{document}